\documentclass[runningheads]{llncs}
\usepackage{graphicx}
\usepackage{ifpdf,latexsym,graphicx,url}
\usepackage[utf8]{inputenc}
\usepackage{subcaption}
\usepackage{ mathtools, amsfonts, grffile, wasysym, xspace, amsmath, float, ifthen, xifthen,thmtools, thm-restate}
\usepackage[table,xcdraw]{xcolor}

\usepackage{todonotes}

\usepackage{tikz,tikz-cd}
\usetikzlibrary{shapes.geometric}

\newcommand{\ops}{\ensuremath{\mathsf{ops}}}

\newcommand{\SG}[3]{\ensuremath{(#1,#2)}\textsc{-AEC-#3}}
\newcommand{\SGSTD}[2]{\SG{#1}{#2}{Std}}

\newcommand{\SGNP}[2]{\SG{#1}{#2}{NP}}
\newcommand{\SGEP}[2]{\SG{#1}{#2}{EP}}

\newcommand{\Partition}{\textsc{Partition}\xspace}

\newcommand{\ProductPartition}{\textsc{ProductPartition}\xspace}

\newcommand{\ProductPartitionEqual}{\textrm{\ProductPartition-\texorpdfstring{$n/2$}}\xspace}

\newcommand{\TPartition}{\textsc{3-Partition}\xspace}
\newcommand{\TPartitionEqual}{\textsc{3-Partition}\-3\xspace}

\newcommand{\defineproblemrefcomment}[5]{%
  \begin{problem}[\unboldmath #1]\rm~

  \noindent\textbf{Instance:} #2

  \noindent\textbf{Question:} #3
  
  \noindent\textbf{Reference:} #4
  
  \noindent\textbf{Comment:} #5
  \end{problem}%
}

\newcommand{\defineproblem}[3]{%
  \begin{problem}[\unboldmath #1]\rm~

  \noindent\textbf{Instance:} #2

  \noindent\textbf{Question:} #3

  \end{problem}%
}

\begin{document}
\title{The Computational Complexity of Finding Arithmetic Expressions With and Without Parentheses}
\titlerunning{Finding Arithmetic Expressions With and Without Parentheses}
%
\author{Jayson Lynch\inst{1}
\and
Yan (Roger) Weng\inst{2}}
\authorrunning{}
%
\institute{University of Waterloo, Waterloo, ON, Canada, \and
The Peddie School, 201 S Main St, Hightstown, NJ 08520, USA}
\maketitle              
\begin{abstract}
We show NP-completeness for various problems about the existence of arithmetic expression trees. When given a set of operations, inputs, and a target value does there exist an expression tree with those inputs and operations that evaluates to the target? We consider the variations where the structure of the tree is also given and the variation where no parentheses are allowed in the expression.

\keywords{NP-completeness \and Arithmetic Expression Trees \and Computational Complexity}
\end{abstract}
\section{Introduction}

Arithmetic expression trees are trees with numbers at the leaves and operators at internal nodes. Each operator takes its children as inputs and provides the evaluation to its parent. An example can be seen in Figure~\ref{fig:exptree}. These can be seen as a restricted form of arithmetic circuits which form a directed acyclic graph whose nodes have operators. Arithmetic and symbolic expression trees are important structures in computer science.

\definecolor{op}{rgb}{.72,.64,.77}
\definecolor{leaf}{rgb}{.98,.91,.52}

\begin{figure}[H] 
  \centering
  \begin{tikzpicture}[
    every node/.style={minimum width=1.8em,inner sep=0,draw,circle},
    level distance=0.8cm,
    level 1/.style={sibling distance=3cm},
    level 2/.style={sibling distance=1.5cm},
    level 3/.style={sibling distance=1.5cm},
  ]
  \node[label=98,fill=op] at (0,0) {$\boldsymbol-$}
    child { node[label=99,fill=op] {$\boldsymbol\times$}
      child { node[fill=leaf] {$11$} }
      child { node[fill=leaf] {$9$} }
    }
    child { node[label=1,fill=op] {$\boldsymbol\div$}
      child { node[fill=leaf] {$4$} }
      child { node[label=4,fill=op] {$\boldsymbol+$}
        child { node[fill=leaf] {$3$} }
        child { node[fill=leaf] {$1$} }
      }
    };
  \end{tikzpicture}
  \caption[Arithmetic Expression Construction results]
    {An example expression tree for $9 \times 11 - (4 \div (3 + 1)) = 98$. The numbers above the internal nodes indicate their values.}
  \label{fig:exptree}
\end{figure}
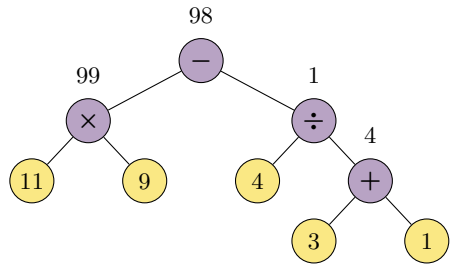

Recently arithmetic expression trees have been used in machine learning systems making progress in solving math word problems. Despite their success at many natural language processing tasks, it was shown that transformers and other large language models like GPT do not perform well on math word problems and mathematical reasoning\cite{hendrycks2020measuring,hendrycks2021measuring}. A number of systems have used sequence-to-tree and graph-to-tree neural networks to generate arithmetic expression trees from input questions and used those to calculate answers. The main targets have been the Math23k and MAWPS data sets which contain more than 23,000 math word problems\cite{koncel2016mawps,li2019modeling,wang2019template,zhang2020graph2tree,li2020graph2tree,wu2020knowledge,qin2020semantically,lan2021mwptoolkit}. Similar techniques have been used to tackle more advanced problems. Prediction of arithmetic expression trees using transformers\cite{vaswani2017attention} and graph neural networks has also been used to solve machine learning exercises from MIT's 6.036 Introduction to Machine Learning\cite{tran2021solving}. Predicting symbolic expression trees has also been applied to solving integration and ordinary differential equation \cite{lample2019deep} problems and outperforms Mathematica and Maple on test problems.  

The need to predict arithmetic expression trees with constraints motivates the further study of the computational complexity of the problem. Versions of this Arithmetic Expression Tree problem were studied in \cite{alcock2020arithmetic} motivated by algebraic complexity theory and recreational math. The construction of small algebraic circuits has been given significant study with one of the earliest being Scholz's 1937 study of
minimal addition chains \cite{scholzchains}, which is equivalent to finding
the smallest circuit with operation $+$ that outputs a target value~$t$. This problem was later shown to be NP-complete \cite{minadditionchains}. 

From the recreational side, arithmetic expression construction can be seen as a generalization of a number of games and educational puzzles. Perhaps best known is The 24 Game, a card game involving finding sets of four cards that can be combined with standard operations to equal 24 \cite{24game-wiki}. A variation of the problem also shows up in the Number Round of the British game show Countdown; however in this version the target is randomly generated and contestants do not need to use all of the numbers given. Some combinatorial and algorithmic aspects of this specific problem have been studied in other papers\cite{ozkul27finding,alliot2015final,colton2014countdown}.


Leo et al \cite{alcock2020arithmetic} define the main problem of Arithmetic Expression Construction:

\defineproblem{\SGSTD{\mathbb{L}}{\ops} / \textrm{Standard}}
{A multiset of values $A = \{a_1, a_2, \dots, a_n\}\subseteq \mathbb{L}$ and a target value $t\in\mathbb{L}$.}
{Does there exist a parenthesized expression using any of the operations in $\ops$ containing each element of $A$ exactly once and evaluating to $t$?}

Here as in the prior paper we will take $\ops$ to be a subset of $\{+,-,\times,\div\}$.

Alcock et al. \cite{alcock2020arithmetic} also defines two variants, one where all operations in the tree are given and the question is whether the input numbers can be assigned to leaves of that tree to evaluate to the target, and the other fixes the order of the numbers in the expression, but leaves the choice of operator open. The paper gives algorithms and NP-completeness results for many subsets of operations for these problems. It provides a full characterization for the enforced operations version, and and a full characterization for the standard version up to weak versus strong NP-completeness. Containment in NP is trivial as an expression tree can be evaluated in polynomial time, giving a polynomial time witness.

We explore two new versions of this problem: 1) enforced parentheses gives the structure of the tree which the expression tree must conform to, alternatively it specifies a full parenthesization of the expression. 2) no parentheses in which operations must be evaluated with respect to the standard order of operations. It is interesting to note this second version is equivalent to depth $2$ trees whose operators have arbitrary fan-in rather than being binary operators. In Section~\ref{sec:Enforced Parentheses} we show the reductions for standard arithmetic expression construction carry over giving the same characterization. In Section~\ref{sec:No Parentheses} we show a number of NP-completeness results detailed in Table~\ref{table:summary}.



We consider the following two variants of AEC which impose additional constraints (specified by some data we denote by $D$) on the expression trees:




\defineproblem{\SGEP{\mathbb{L}}{\ops} / \textrm{Enforced Parenthesis}}
{A multiset of values $A = \{a_1, a_2, \dots, a_n\}\subseteq \mathbb{L}$, a target value $t\in\mathbb{L}$, and a tree $D$.}
{Does there exist a parenthesized expression using any of the operations in $\ops$ that contains each element of $A$ exactly once and evaluates to $t$ such that its expression tree is isomorphic to $D$?}

\defineproblem{\SGNP{\mathbb{L}}{\ops} / \textrm{No Parenthesis}}
{A multiset of values $A = \{a_1, a_2, \dots, a_n\}\subseteq \mathbb{L}$ and a target value $t\in\mathbb{L}$.}
{Does there exist an no-parenthesis expression using any of the operations in $\ops$ that contains each element of $A$ exactly once and evaluates to $t$? (Multiplications and divisions are done before additions and subtractions)}

\begin{table}[htbp]
\centering
\footnotesize
\tabcolsep=2pt
\begin{minipage}{\textwidth}%
\centering
\begin{tabular}{|c|c|c|} \hline
\textbf{Operations}  &  \textbf{Enforced Parenthesis (\textsection \ref{sec:Enforced Parentheses})}  &  \textbf{No Parenthesis (\textsection \ref{sec:No Parentheses})}  \\ \hline
$\{+\}$  &  $\in \text{P}$   &  $\in \text{P}$ (\textsection \ref{sec:single})  \\ \hline
$\{-\}$  &  weakly NP-complete   &  $\in \text{P}$ (\textsection \ref{sec:single}) \\ \hline
$\{\times\}$  &  $\in \text{P}$   &  $\in \text{P}$ (\textsection \ref{sec:single})   \\ \hline
$\{\div\}$  &  strongly NP-complete   &  $\in \text{P}$ (\textsection \ref{sec:single})  \\ \hline
$\{+,-\}$  &  weakly NP-complete   &  {weakly NP-complete} (\textsection \ref{sec:plusminus}) \\ \hline
$\{+,\times\}$  &  {weakly NP-complete}  &  weakly NP-complete  (\textsection \ref{sec:plustimes}) \\ \hline
$\{+,\div\}$  &  {weakly NP-complete }   &  {weakly NP-complete } (\textsection \ref{sec:plusdivide})\\ \hline
$\{-,\times\}$  &  weakly NP-complete   &  {weakly NP-complete}   (\textsection \ref{sec:minustimes}) \\ \hline
$\{-,\div\}$  &  {weakly NP-complete}   &  {weakly NP-complete} (\textsection \ref{sec:minusdivide}) \\ \hline
$\{\times,\div\}$  &  {strongly NP-complete}   &  {strongly NP-complete}  (\textsection \ref{sec:timesdivide}) \\ \hline
$\{+,-,\times\}$  &  {  weakly NP-complete }  &  {weakly NP-complete}   (\textsection \ref{sec:plusminustimes}) \\ \hline
$\{+,-,\div\}$  &  {weakly NP-complete}   &  {weakly NP-complete} (\textsection \ref{sec:minusdivide})\\ \hline
$\{+,\times,\div\}$  &  {weakly NP-complete }    &  {weakly NP-complete}   (\textsection \ref{sec:plustimesdivide}) \\ \hline
$\{-,\times,\div\}$  &  {    weakly NP-complete }  &  {weakly NP-complete} (\textsection \ref{sec:minusdividetimes})\\ \hline
$\{+,-,\times,\div\}$  &  {weakly NP-complete }  &  {weakly NP-complete} (\textsection \ref{sec:minusdividetimes})\\ \hline
\end{tabular}

\caption{Results for all new variations of arithmetic expression construction.}

\label{table:summary}

\end{minipage}

\end{table}

\section{No Parentheses}
\label{sec:No Parentheses}

In this section we provide a number of proofs relying on the Rational Function Framework from \cite{alcock2020arithmetic}. In their paper they introduce the arithmetic expression construction over ratios of polynomials and prove there is a polynomial time reduction from the problem over integers to the problem over ratios of polynomials. The proof shows how to construct integers which are sufficiently large that they cannot be generated by any expression comprised of standard arithmetic operations and a given multi-set of integers. Further, this constructed number is no more than exponential in the size and number of those integers. The basis representation theorem then allows the construction of a reduction between the proposed Arithmetic Expression Construction problems over integers and over rational polynomials. The additional restrictions of enforced parentheses has no impact on the proof of the Rational Function Framework and thus it suffices to show NP-hardness for the polynomial versions of the arithmetic construction problems we consider. Most of the reductions given show hardness over the field of rational polynomials which implies hardness for the problem over integers by a simple adaptation of Theorem 2.1 in \cite{alcock2020arithmetic}. Although working with polynomials significantly simplifies our proofs, one drawback is this means most of our results only show weak NP-hardness.

\subsection{$\{+,/,*\}$ no parenthesis is weakly NP-hard}
\label{sec:plustimesdivide}

To prove that $(\mathbb{N}[x,y], \{+,/,*\})$ -AEC-NP is weakly NP-hard, we reduce from product partition-n/2\cite{alcock2020arithmetic}. Given an instance of product partition-n/2, \\ $A = \{a_1, a_2,\cdots , a_n\}$, construct instance $I_A$ of $(\mathbb{N}[x,y], \{+,/,*\})$-AEC-NP with the set of values \\ $\{a_1x,\underset{i\ge 2}{a_i}xy,y^{\frac{n}{2}-1},y^{\frac{n}{2}}\}$, target $t=2x^{\frac{n}{2}}\sqrt{\prod_ia_i}$.\\\\
If there is a solution for the product partition-n/2, we can divide A into two subsets $A_1$ and $A_2$ of equal size such that $\frac{\prod A_1}{\prod A_2}=1$. We can partition $a_0x, a_ixy$ into corresponding sets and take their products to get the polynomial $x^{\frac{n}{2}}y^{\frac{n}{2}}\sqrt{\prod_ia_i}$ and another polynomial $x^{\frac{n}{2}}y^{\frac{n}{2}-1}\sqrt{\prod_ia_i}$. We can then construct an instance that gives us the target value: $$\frac{x^{\frac{n}{2}}y^{\frac{n}{2}}\sqrt{\prod_ia_i}}{y^{
\frac{n}{2}}}+\frac{x^{\frac{n}{2}}y^{\frac{n}{2}-1}\sqrt{\prod_ia_i}}{y^{
\frac{n}{2}-1}}=2x^{\frac{n}{2}}\sqrt{\prod_ia_i}.$$

\begin{lemma} \label{+numbers}
In $\{+,*\}$ and $\{+,/,*\}$ with no parenthesis, if the set of values are polynomials with positive integer coefficients that include exactly $n$ $x^1$ terms and the target only has a $x^{\frac{n}{2}}$ term, then there is at most one  ``$+$'' in the arithmetic construction that would yield the target value. 
\end{lemma} 
\begin{proof}

Suppose that there are two or more ``$+$"s in the construction, then the instance can be expressed as $P_1+P_2+\cdots+P_n (n\ge 3)$, where $P_i$ represents a part of our construction that only uses the ``$*$" and ``$/$" operations.

The target value only consists of a $x^\frac{n}{2}$ term. By the rule of addition and because everything is positive, it must be true that each of $P_i$ is a $x^\frac{n}{2}$ term.

Now we show that each $P_i$ consists at least $\frac{n}{2}$ $x^1$ terms among the set of values that are given. The rule of multiplication tells us that if we multiply $\frac{n}{2}$ $x^1$ terms together, we will have a $x^\frac{n}{2}$ term. If we use division at any point, the degree of $x$ would only be smaller. Therefore, there are at least $\frac{n}{2}$ $x^1$ terms in $P_i (1\le i\le n)$.

Since $n\ge 3$, we need to use at least $\frac{3n}{2}$ $x^1$  terms in our construction, but we only have $n$ such terms that are given. Thus, it is impossible to have two or more ``$+$"s in our construction.
\end{proof}

Now we prove the converse. Multiplication and division are done before addition. The target only has a $x^{\frac{n}{2}}$ term. By Lemma \ref{+numbers}, there must be $0$ or $1$ "+"s in our construction.
Therefore, we have to either sum up two $x^{\frac{n}{2}}$ terms, or divide a $x^{\frac{3n}{4}}$ term by a $x^{\frac{n}{4}}$ term.
If we divide a $x^{\frac{3n}{4}}$ term by a $x^{\frac{n}{4}}$ term, then we have $\frac{\text{product of } (\frac{3}{4}n) a_i xy \text{ terms}}{\text{product of } (\frac{1}{4}n-1) a_i xy \text{ terms}\cdot a_0x}$ or $\frac{\text{product of } (\frac{3}{4}n-1)a_i xy \text{ terms}\cdot a_0x}{\text{product of } (\frac{1}{4} n)a_i xy \text{ terms}}$. Therefore, the degree of $y$ is either $\frac{n}{2}-1$ or $\frac{n}{2}+1$. We are left with two other elements in our set, which are $y^{\frac{n}{2}-1},y^{\frac{n}{2}}$. However, it is impossible to cancel out $y^{\frac{n}{2}+1}$ (or $y^{\frac{n}{2}-1}$) when we need both $y^{\frac{n}{2}-1}$ and $y^{\frac{n}{2}}$ in the fraction, so the result must include a term with $y$, which means that it cannot be our target value.
Therefore, the only possibility is to sum up two $x^{\frac{n}{2}}$ terms.

To do so, we need to sum up the product of $(\frac{n}{2}) \hspace{2mm} a_ixy$ terms and the product of $(\frac{n}{2}-1) \hspace{2mm} a_ixy$ terms and the $a_1x$ term.

Suppose that the product of the $(\frac{n}{2}) \hspace{2mm} a_ixy$ terms is $t\cdot x^{\frac{n}{2}} y^{\frac{n}{2}}$ ($t$ is an integer), and the product of the $(\frac{n}{2}-1) \hspace{2mm} a_ixy$ terms and the $a_1x$ term is $s\cdot x^{\frac{n}{2}}y^{{\frac{n}{2}}-1}$ ($s$ is an integer). We know that $t\cdot s=\prod_i a_i$. If the solution is also valid for our instance $I_A$ of $(\mathbb{N}[x,y], \{+,/,*\})$ -AEC-NP, then $t+s=2\sqrt{\prod_i a_i}$.
Combining the two equations, we have
\[t\cdot s=\prod_i a_i \]
\[t+s=2\sqrt{\prod_i a_i}\]
Let $k=\frac{t}{\sqrt{\prod_i a_i}}$, and $m=\frac{s}{ \sqrt{\prod_i a_i}}$.
We have 
\[k\cdot m=1 \]
\[k+m=2\]
The only real solution to this system of equations is $k=m=1$. Therefore, $t=s=\sqrt{\prod_ia_i}$, which means that the product of $(\frac{n}{2}) \hspace{2mm} a_ixy$ terms is $\sqrt{\prod_ia_i}x^{\frac{n}{2}}y^{\frac{n}{2}}$, and the product of the $(\frac{n}{2}-1) \hspace{2mm} a_ixy$ terms and the $a_1x$ term is $\sqrt{\prod_ia_i}x^{\frac{n}{2}}y^{\frac{n}{2}-1}$. 
\\

We now have $x^{\frac{n}{2}}y^{\frac{n}{2}}\sqrt{\prod_ia_i}+x^{\frac{n}{2}}y^{\frac{n}{2}-1}\sqrt{\prod_ia_i}$. We are left with $y^{\frac{n}{2}-1},y^{\frac{n}{2}}$, and the only way to cancel out the $y$ term is by dividing the first product by $y^{\frac{n}{2}}$ and dividing the second product by $y^{\frac{n}{2}-1}$. Thus, the problem can only be solved if there is a solution to product partition-n/2.

\subsection{$\{+,*\}$ no parenthesis is weakly NP-hard}
\label{sec:plustimes}

To prove that $(\mathbb{N}[x], \{+, *\})$ -AEC-NP is weakly NP-hard, we proceed by reduction from product partition-n/2. Given an instance of product partition-n/2, $A = \{a_1, a_2,\cdots , a_n\}$, construct instance $I_A$ of $(\mathbb{N}[x], \{+, *\})$-AEC-NP with the set of values $\{a_1x,a_2x,\cdots,a_nx\}$, target $t=2x^{\frac{n}{2}}\sqrt{\prod_ia_i}$.

If this instance of product partition-n/2 has a solution, we can divide A into two subsets $A_1$ and $A_2$ of equal size such that $\frac{\prod A_1}{\prod A_2}=1$. We can partition $a_ix$ into corresponding sets and take their products to get two polynomials of value $x^{\frac{n}{2}}\sqrt{\prod_ia_i}$. We can add these two polynomials, which would give us the target: $t=2x^{\frac{n}{2}}\sqrt{\prod_ia_i}$.

Now we prove the converse. The target only has a $x^{\frac{n}{2}}$ term and the only operations we are allowed to use are $*$ and $+$. Thus each group of multiplications must have exactly $n/2$ terms by Lemma \ref{+numbers}. The coefficient of $x^{\frac{n}{2}}$ is $2\sqrt{\prod_ia_i}$, so the problem is the same as product partition-n/2.

\subsection{$\{-,*\}$ no parenthesis is weakly NP-hard}
\label{sec:minustimes}
To show that $(\mathbb{N}[x,y], \{-,*\})$ -AEC-NP is weakly NP-hard, we reduce from product partition-n/2. Given an instance of product partition, $A = \{a_1, a_2,\cdots , a_n\}$, construct instance $I_A$ of $(\mathbb{N}[x,y], \{-,*\})$ -AEC-NP with the set of values \\ $\{a_1x,a_2x,\cdots,a_nx,y,y\}$, target $t=0$.

If the product partition-n/2 has a solution, we divide A into two subsets $A_1$ and $A_2$ of equal size such that $\frac{\prod A_1}{\prod A_2}=1$. We can partition $a_ix$ into corresponding sets and take their products to get two polynomials of value $x^{\frac{n}{2}}\sqrt{\prod_ia_i}$. Then we can construct an instance with the remaining two copies of $y$ to give the target value: $x^{\frac{n}{2}}\sqrt{\prod_ia_i}\cdot y - x^{\frac{n}{2}}\sqrt{\prod_ia_i}\cdot y=0$.

Now we show that this is the only solution to $(\mathbb{N}[x,y], \{-,*\})$ -AEC-NP. All expressions of $(\mathbb{N}[x,y], \{-,*\})$ -AEC-NP can be expressed as $P_1-P_2-\cdots-P_n=0$, where $P_i$ represents a part of our construction that only uses the "$*$" operation. We can rewrite this as $P_1=P_2+\cdots+P_n$

First of all, we can show that $P_1$ must include exactly one copy of $y$. 
If $P_1$ includes $y^2$, then $P_1\neq P_2+\cdots+P_n$, because $P_2+\cdots+P_n$ would not have any  terms with $y$. 

If $P_1$ does not contain $y$, then two copies of $y$ must be in $P_i$ and $P_j$, where $i,j\ge 2$. In the rational function framework, we assumed that $x$ is a large positive integer, and all $a_i$ are natural numbers, so $P_i+P_j>0$ and it contains a term with $y$. However, $P_1$ does not contain $y$, so $P_1\neq P_2+\cdots+P_n$. To sum up, $P_1$ and another $P_i$ both consist $y$. Without loss of generality, we assume that $P_1$ and $P_2$ each has one copy of $y$.

Now we can prove that $n=2$. In other words, the only solution is $P_1=P_2$. Suppose that $n\ge 3$, then $P_1=P_2+\cdots+P_n$, with $P_1$ and $P_2$ each includes $y$. We have $P_1-P_2=P_3+\cdots+P_n$.

If $P_1-P_2=0$, then $P_3+\cdots+P_n=0$. This is impossible because each of $a_ix$ is a positive number. 

If $P_1-P_2\neq 0$, then $P_1-P_2$ includes $y$. However, $P_3+\cdots+P_n$ does not include $y$, so this is contradictory. 

Therefore, it is not possible that $n\ge 3$, which means $P_1=P_2$, and $P_1, P_2$ both include $y$. Since $P_1=P_2$, $\deg_x(P_1)=\deg_x(P_2)=x^{\frac{n}{2}}$. Thus, $P_1$ and $P_2$ each include half of the $a_ix$ elements, and their products are equal to each other. This problem is the same as the product partition-n/2.

\subsection{$\{+,-,*\}$ no parenthesis is weakly NP-hard}
\label{sec:plusminustimes}

To prove that $(\mathbb{N}[x,y], \{+,-,*\})$ -AEC-NP is weakly NP-hard, we reduce from $3$-partition-3\cite{GareyJohnson}. Given an instance of $3$-partition-3, $A = \{a_1, a_2,\cdots , a_n\}$, construct instance $I_A$ of $(\mathbb{N}[x,y], \{+,-,*\})$-AEC-NP with the set of values $\{\underset{1\le i \le n}{a_ix}\}$ and three copies of $\{\underset{1\le j \le \frac{n}{3}}{y_j}\}$. The target $t=(\frac{3}{n}\sum a_i)(\sum y_j)x$.

If this instance of $3$-partition-3 has a solution, we can divide A into subsets $A_1,A_2,\cdots,A_\frac{n}{3}$ of equal size and equal sum. The sum of elements in each subset is $\frac{3}{n}\sum_{1\le i \le n} a_i$. We can partition $a_ix$ into corresponding sets $A_i$, and multiply each of the element in the same set, $A_i$, by $y_i$, and add all such products together. We can construct an instance that gives us the target value: $\sum_{a_i \in A_1} a_ixy_1 + \sum_{a_i \in A_2} a_ixy_2 + \cdots \sum_{a_i \in A_{\frac{n}{3}}xy} a_ixy_{\frac{n}{3}}=(\frac{3}{n}\sum a_i)(\sum y_j)x$.

Now we prove the converse. Suppose that the instance $I_A$ of $(\mathbb{N}[x,y], \{+,-,*\})$-AEC-NP has a solution. Multiplications are done before additions and subtractions, so a solution to this problem can be expressed as $P_1\pm P_2\pm \cdots \pm P_m$. Our target has a $xy_1$ term, a $xy_2$ term, $\cdots$, and a $xy_\frac{n}{3}$ term. Since there are $3$ copies of each $y_i$, the solution must be $(\text{sum/difference of three} \hspace{2mm} a_ix \hspace{2mm} \text{terms})\cdot y_1 + \cdots + (\text{sum/difference of three} \hspace{2mm} a_ix \hspace{2mm} \text{terms})\cdot y_\frac{n}{3}=k_1xy_1+k_2xy_2+\cdots+k_\frac{n}{3}xy_\frac{n}{3}$ where $k_1,k_2,\cdots,k_\frac{n}{3}$ are constants. In the target, $k_1=k_2=\cdots=k_\frac{n}{3}=\frac{3}{n}\sum a_i$, which means $k_1+k_2+\cdots+k_\frac{n}{3}=\sum a_i$. Therefore, ``$-$" is never used in the instance, and the solution has to be $(\text{sum of three} \hspace{2mm} a_ix \hspace{2mm} \text{terms})\cdot y_1 + \cdots + (\text{sum of three} \hspace{2mm} a_ix \hspace{2mm} \text{terms})\cdot y_\frac{n}{3}$. This problem is the same as finding a solution for the $3$-partition-3 problem.

\subsection{$\{+,/\}$ no parenthesis is weakly NP-hard}
\label{sec:plusdivide}
We reduce from product partition-n/2. Given an instance of product partition-n/2, $A = \{a_1, a_2,\cdots , a_n\}$, construct instance $I_A$ of $(\mathbb{N}[x,y], \{+,/\})$ -AEC-NP with the set of values \\ $\{a_1xy, \underset{i\ge 2}{a_i}x,\sqrt{\prod_ia_i}x^{\frac{n}{2}},\sqrt{\prod_ia_i}x^{\frac{n}{2}}y\}$ and target $t=2$.

If the product partition-n/2 has a solution, we divide A into two subsets $A_1$ and $A_2$ of equal size such that $\frac{\prod A_1}{\prod A_2}=1$. We can partition $a_1xy, \underset{i\ge 2}{a_i}x$ into corresponding sets. We divide $\sqrt{\prod_ia_i}x^{\frac{n}{2}}y$ by all elements in the set that contains $a_1xy$, and then we divide $\sqrt{\prod_ia_i}x^{\frac{n}{2}}$ by all elements in the other set. Finally, we add up the two polynomials to get the target value: $\frac{x^{\frac{n}{2}}\sqrt{\prod_ia_i}}{x^{\frac{n}{2}}\sqrt{\prod_ia_i}}+\frac{x^{\frac{n}{2}}y\sqrt{\prod_ia_i}}{x^{\frac{n}{2}}y\sqrt{\prod_ia_i}}=2 $

Now we prove that this is the only solution to this instance $I_A$ of $(\mathbb{N}[x,y], \{+,/\})$ -AEC-NP. Division are done before additions. All expressions of $(\mathbb{N}[x,y], \{+,/\})$ -AEC-NP can be expressed as $P_1+P_2+\cdots+P_n=2$, where $P_i$ represents a part of our construction that only uses the "$/$" operation.
$2$ is a constant, so $P_1, P_2, \cdots, P_n$ must all be constants (because $a_i\ge 0$).

We can show that the two elements $\sqrt{\prod_ia_i}x^{\frac{n}{2}}$ and $\sqrt{\prod_ia_i}x^{\frac{n}{2}}y$ are not in the same $P_i$. Suppose that these two elements are in the same $P_i$, and the value of $P_i$ is a constant. $\sqrt{\prod_ia_i}x^{\frac{n}{2}}y$ must be the dividend due to the fact that it is the element of the set that has the largest degrees for $x$ and $y$. (If $\sqrt{\prod_ia_i}x^{\frac{n}{2}}y$ is a divisor for some $P_i$, then the denominator of the resulting fraction of $P_i$ must contain $x$ or $y$. In this case, $P_i$ is no longer a constant.) We assumed that $\sqrt{\prod_ia_i}x^{\frac{n}{2}}$ is one of the divisors, and we know that we need to include $a_1xy$ as another divisor to cancel out the $y$ term in $\sqrt{\prod_ia_i}x^{\frac{n}{2}}y$. However, after dividing $a_1xy$ and $\sqrt{\prod_ia_i}x^{\frac{n}{2}}$, the degree of $x$ is $-1$, and $/$ is the only operation allowed. It is impossible for $P_i$ to be a constant, which means that $\sqrt{\prod_ia_i}x^{\frac{n}{2}}$ and $\sqrt{\prod_ia_i}x^{\frac{n}{2}}y$ are not in the same $P_i$. 

We can assume that $\sqrt{\prod_ia_i}x^{\frac{n}{2}}$ is in $P_1$ and $\sqrt{\prod_ia_i}x^{\frac{n}{2}}y$ is in $P_2$. $\sqrt{\prod_ia_i}x^{\frac{n}{2}}$ and $\sqrt{\prod_ia_i}x^{\frac{n}{2}}y$ must both be the dividends in $P_1$ and $P_2$, because all the other elements have lower degrees of $x$. $P_1$ is a constant, so $\sqrt{\prod_ia_i}x^{\frac{n}{2}}$ must be divided by $(\frac{n}{2}) \hspace{2mm} a_ix$ terms. $P_2$ is also a constant, so $\sqrt{\prod_ia_i}x^{\frac{n}{2}}y$ must be divided by $(\frac{n}{2}-1) \hspace{2mm} a_ix$ terms and the $a_1xy$ term.

Suppose that the product of the $(\frac{n}{2}) \hspace{2mm} a_ix$ terms in $P_1$ is $t\cdot x^{\frac{n}{2}}$ ($t$ is an integer), and the product of the $(\frac{n}{2}-1) \hspace{2mm} a_ix$ terms and the $a_1xy$ term in $P_2$ is $s\cdot x^{\frac{n}{2}}y$ ($s$ is an integer). We know that $t\cdot s=\prod_i a_i$. If the solution is also valid for our instance $I_A$ of $(\mathbb{N}[x,y], \{+,/\})$ -AEC-NP, then $\frac{\sqrt{\prod_ia_i}}{t}+\frac{\sqrt{\prod_ia_i}}{s}=2$.
Combining the two equations, we have
\[t\cdot s=\prod_i a_i \]
\[\frac{\sqrt{\prod_ia_i}}{t}+\frac{\sqrt{\prod_ia_i}}{s}=2\]
Let $k=\frac{\sqrt{\prod_ia_i}}{t}$, and $m=\frac{\sqrt{\prod_ia_i}}{s}$.
We have 
\[k\cdot m=1 \]
\[k+m=2\]
The only real solution to this system of equations is $k=m=1$. Therefore, $t=s=\sqrt{\prod_ia_i}$, which means that the product of $(\frac{n}{2}) \hspace{2mm} a_ix$ terms in $P_1$ is $\sqrt{\prod_ia_i}x^{\frac{n}{2}}$, and the product of the $(\frac{n}{2}-1) \hspace{2mm} a_ix$ terms and the $a_1xy$ term in $P_2$ is $\sqrt{\prod_ia_i}x^{\frac{n}{2}}y$. 

\subsection{$\{-,/\}$, $\{+,-,/\}$ no parenthesis are weakly NP-hard}
\label{sec:minusdivide}
To show that $(\mathbb{N}[x], \{+,-,/\})$ -AEC-NP is weakly NP-hard, we reduce from product partition-n/2. Given an instance of product partition-n/2, $A = \{a_1, a_2,\cdots , a_n\}$, construct instance $I_A$ of $(\mathbb{N}[x], \{+,-,/\})$ -AEC-NP with the set of values \\ $\{\underset{i\ge 1 }{a_i}x,\sqrt{\prod_ia_i}x^{2n},\sqrt{\prod_ia_i}x^n\}$, target $t=x^{\frac{3n}{2}}-x^{\frac{n}{2}}$.

If the product partition-n/2 has a solution, we divide A into two subsets $A_1$ and $A_2$ of equal size such that $\frac{\prod A_1}{\prod A_2}=1$. We can partition $a_ix$ into corresponding sets. We divide $\sqrt{\prod_ia_i}x^{2n}$ by $\prod A_1 \cdot x^\frac{n}{2}$, and then we divide $\sqrt{\prod_ia_i}x^n$ by $\prod A_2 \cdot x^\frac{n}{2}$. Finally, we subtract the second polynomial from the first polynomial to get the target value: $\frac{\sqrt{\prod_ia_i}x^{2n}}{\sqrt{\prod_ia_i}x^{\frac{n}{2}}}-\frac{\sqrt{\prod_ia_i}x^{n}}{\sqrt{\prod_ia_i}x^{{\frac{n}{2}}}}=x^{\frac{3n}{2}}-x^{\frac{n}{2}}$.

Now we show that this is the only possible solution. All expressions of $(\mathbb{N}[x], \{+,-,/\})$ -AEC-NP can be expressed as $P_1\pm P_2\pm \cdots \pm P_m=x^{\frac{3n}{2}}-x^{\frac{n}{2}}$, where $P_i$ represents a part of our construction that only uses the "$/$" operation. 

The target value has a $x^{\frac{3n}{2}}$ term, which means that at least one of the $P_i$ is $c\cdot x^{\frac{3n}{2}}$, where $c$ is any constant. Among all elements given, $\sqrt{\prod_ia_i}x^{2n}$ is the only one that has degree larger than $x^{\frac{3n}{2}}$, which means it must be the dividend of one of the $P_i$. Without loss of generality, let $\sqrt{\prod_ia_i}x^{2n}$ be the dividend of $P_1$. Furthermore, since it is impossible for any of $P_2, P_3,\cdots,P_m$ to have a $x^{\frac{3n}{2}}$ term, $P_1=x^{\frac{3n}{2}}$.

If the element $\sqrt{\prod_ia_i}x^n$ is also in $P_1$, then $\frac{\sqrt{\prod_ia_i}x^{2n}}{\sqrt{\prod_ia_i}x^n}=x^n$. This means that the degree of $P_1$ has to be smaller or equal to $n$, which does not satisfy $P_1=x^{\frac{3n}{2}}$. Therefore, the only possible way is $P_1=\frac{\sqrt{\prod_ia_i}x^{2n}}{\text{product of} \hspace{2mm} (\frac{n}{2}) \hspace{2mm} a_ix \hspace{2mm} \text{terms}}=x^{\frac{3n}{2}}$. This means that the product of the $(\frac{n}{2}) \hspace{2mm} a_ix$ terms in $P_1$ is $\sqrt{\prod_ia_i}x^{\frac{n}{2}}$. 

Now we need a $-x^{\frac{n}{2}}$ term to satisfy the target value. The elements that have left are $\sqrt{\prod_ia_i}x^n$ and $(\frac{n}{2}) \hspace{2mm} a_ix$ terms not used in $P_1$. $\sqrt{\prod_ia_i}x^n$ is the only term that has degree larger than $\frac{n}{2}$, so it has to be the dividend of $P_2$. Since $P_2$ has degree $\frac{n}{2}$, the $(\frac{n}{2}) \hspace{2mm} a_ix$ terms not used in $P_1$ must all be the divisors in $P_2$. We know that $\prod A_1=\prod A_2= \sqrt{\prod_ia_i}$, so $P_2=\frac{\sqrt{\prod_ia_i}x^n}{\sqrt{\prod_ia_i}x^{\frac{n}{2}}}=x^{\frac{n}{2}}$.

This means that $P_1-P_2=x^{\frac{3n}{2}}-x^{\frac{n}{2}}$, and this is the only possible construction for this instance of $(\mathbb{N}[x], \{+,-,/\})$-AEC-NP.

Since $(\mathbb{N}[x], \{-,/\})$-AEC-NP is more restrictive than $(\mathbb{N}[x], \{+,-,/\})$-AEC-NP a correct solution to $(\mathbb{N}[x], \{-,/\})$-AEC-NP implies a correct solution to $(\mathbb{N}[x], \{+,-,/\})$-AEC-NP. Since our yes instances of $(\mathbb{N}[x], \{+,-,/\})$-AEC-NP were constructed to only use subtraction and not multiplication, these together imply $(\mathbb{N}[x], \{-,/\})$-AEC-NP must be weakly NP-hard as well.

\subsection{$\{+,-,/,*\}$, $\{-,/,*\}$ no parenthesis are Weakly NP-hard}
\label{sec:minusdividetimes}
Now we show that $(\mathbb{N}[x], \{+,-,/,*\})$ -AEC-NP is weakly NP-hard, we reduce from product partition-n/2. Given an instance of product partition-n/2, $A = \{a_1, a_2,\cdots , a_n\}$, construct instance $I_A$ of $(\mathbb{N}[x], \{+,-,/,*\})$ -AEC-NP with the set of values $\{\underset{i\ge 1}{a_i}x,x^{2n},x^n\}$, and target $t=\sqrt{\prod_ia_i}x^{\frac{5n}{2}}-\sqrt{\prod_ia_i}x^{\frac{3n}{2}}$.

If the product partition-n/2 has a solution, we divide A into two subsets $A_1$ and $A_2$ of equal size such that $\frac{\prod A_1}{\prod A_2}=1$. We can partition $a_ix$ into corresponding sets and take their products to get two polynomials of value $x^{\frac{n}{2}}\sqrt{\prod_ia_i}$. Then we can construct an instance with $x^{2n}$ and $x^n$ to get the target value: $x^{\frac{n}{2}}\sqrt{\prod_ia_i}\cdot x^{2n} - x^{\frac{n}{2}}\sqrt{\prod_ia_i}\cdot x^n=\sqrt{\prod_ia_i}x^{\frac{5n}{2}}-\sqrt{\prod_ia_i}x^{\frac{3n}{2}}$ 

To prove the converse, we show that this is the only solution to this instance of $(\mathbb{N}[x], \{+,-,/,*\})$ -AEC-NP.  All expressions of $(\mathbb{N}[x], \{+,-,/, *\})$ -AEC-NP can be expressed as $P_1\pm P_2\pm \cdots \pm P_m=\sqrt{\prod_ia_i}x^{\frac{5n}{2}}-\sqrt{\prod_ia_i}x^{\frac{3n}{2}}$, where $P_i$ represents a part of our construction that only uses the "$*$", "$/$" operations. 

The target value is $\sqrt{\prod_ia_i}x^{\frac{5n}{2}}-\sqrt{\prod_ia_i}x^{\frac{3n}{2}}$, so there must exist $P_i$ and $P_j$ such that $P_i$ is a $x^{\frac{5n}{2}}$ term and $P_j$ is a $x^{\frac{3n}{2}}$ term. The sum of degrees of $x$ for all given elements is: $1\cdot n+2n+n=4n$. The sum of degrees of $x$ for $x^{\frac{5n}{2}}$ and $x^{\frac{3n}{2}}$ is $4n$. If we use division at any point, the degree of $x$ would only be smaller. Therefore, if a solution is valid, it must not contain division in its construction. Furthermore, $P_1$ must be $\sqrt{\prod_ia_i}x^{\frac{5n}{2}}$ and $P_2$ must be $\sqrt{\prod_ia_i}x^{\frac{3n}{2}}$. $P_1-P_2$ would give us the target value.

The element $x^{2n}$ must be used to construct the $x^{\frac{5n}{2}}$ term for $P_1$, because the degree of $x$ for the product of all the other given elements is still less than $\frac{5n}{2}$. Since division should not be used to construct a solution, we have to multiply $x^{2n}$ by $(\frac{n}{2}) \hspace{2mm} a_i x$ terms to get $\sqrt{\prod_ia_i}x^{\frac{5n}{2}}$. Therefore, the product of the $(\frac{n}{2}) \hspace{2mm} a_i x$ terms in $P_1$ is $\sqrt{\prod_ia_i}$.

The degree of $x$ in $P_2$ is $\frac{3n}{2}$, so the only possible solution is to multiply all the remaining elements, which are $x^n$ and the $(\frac{n}{2})$ remaining  $a_i x$ terms.

Thus, this is the only possible construction, so $(\mathbb{N}[x], \{+,-,/,*\})$ -AEC-NP is weakly NP-hard. 
In our construction, correct instances did not use "$+$", so
$(\mathbb{N}[x], \{-,/,*\})$ -AEC-NP is weakly NP-hard as well. 

\subsection{$\{+,-\}$ is weakly NP-hard}
\label{sec:plusminus}
To show that $(\mathbb{N}, \{+,-\})$ -AEC-NP is weakly NP-hard, we reduce from partition. Given an instance of partition, $A = \{a_1, a_2,\cdots , a_n\}$, construct instance $I_A$ of $(\mathbb{N}, \{+,-,\})$ -AEC-NP with the set A, and target $t=0$.
\\\\
If the partition problem has a solution, then there is a partition $(A_1,A_2)$ of $A$ such that $\sum A_1=\sum A_2$. We can construct an expression of the form $a_{i_1}+a_{i_2}+\cdots+a_{i_k}-(a_{j_1}+a_{j_2}+\cdots+a_{j_{k'}})=0$ where $a_i\in A_1$ and $a_j\in A_2$. All solutions of $(\mathbb{N}, \{+,-\})$ -AEC-NP can be rewritten in this way completing the reduction.

\subsection{$\{*,/\}$ is strongly NP-hard}
\label{sec:timesdivide}
To show that $(\mathbb{N}, \{*,/\})$ -AEC-NP is strongly NP-hard, we reduce from product partition\cite{ProductPartition}. Given an instance of product partition, $A = \{a_1, a_2,\cdots , a_n\}$, construct instance $I_A$ of $(\mathbb{N}, \{*,/\})$ -AEC-NP with the set A, and target $t=1$.

If the product partition problem has a solution, then there is a product partition $(A_1,A_2)$ of $A$ such that $\prod A_1=\prod A_2$. We can construct an expression of the form $\frac{a_{i_1}\cdot a_{i_2}\cdots a_{i_k}}{a_{j_1}\cdot a_{j_2}\cdots a_{j_{k'}}}=1$ where $a_i\in A_1$ and $a_j\in A_2$. All solutions of $(\mathbb{N}, \{*,/\})$ -AEC-NP can be rewritten in this way completing the reduction.

\subsection{$\{+\}$, $\{-\}$, $\{*\}$, $\{/\}$ can be solved in polynomial time}
\label{sec:single}

$\{+\}$ -AEC-NP, $\{*\}$ -AEC-NP $\in P$, because checking the solvability is the same as testing whether the sum (or product) of all elements equals the target value.

For $\{-\}$ -AEC-NP, suppose that the given elements are ${a_1,a_2,\cdots,a_n}$, and the target is $t$. For all $a_i$, we need to check if $a_i-(\sum_{1\le j\le n,j\ne i} a_j)=t$, which can be done in polynomial time.

Similarly, for $\{/\}$ -AEC-NP, we need to check if there exists $a_i$ such that $\frac{a_i}{\prod_{1\le j\le n,j\ne i} a_j}=t$. This can also be done in polynomial time.

\section{Enforced Parenthesis}
\label{sec:Enforced Parentheses}

In this section we show the results of the standard problems carry over to the enforced parenthesis problems for all sets of operations. We can show that the reductions for standard problems mentioned in the paper Arithmetic Expression Construction
would still work when we enforce the expression tree. Since the enforced tree provides a subset of all possible constructions in standards, if the the constructions in the standard problems have a structure that enforces the parenthesization of the problem, then such partitioned structure is also the only possibility in the enforced tree version of the same problem. Thus we simply need to show that if partition problems are solvable, then the enforced parenthesis problems remain solvable. This can be seen by an inspection of the reductions used in that paper; however we give some examples here for guidance.

Example(Standard $\{-\}$): The paper Arithmetic Expression Construction shows that Standard $\{-\}$ is Weakly NP-hard. It shows that if the Partition problem has a solution, we can construct an expression of the form
$(p_1-n_2-\cdots-n_{|A_2|})-(n_1-p_2-\cdots-p_{|A_1|})=\sum A_1-\sum A_2=0$
with $p_i$ in $A_1$ and $n_i$ in $A_2$. Conversely, any solution to the produced $(\mathbb{N},\{-\})$-AEC-Std instance can be factored into this form.

Now we show that this reduction would still work for Enforced Tree $\{-\}$. We can construct an expression of the form $(((p_1-n_2)-n_3)-\cdots-n_{|A_2|})-(((n_1-p_2-)-p_3)-\cdots-p_{|A_1|})$.The solution for this instance exists if the partition problem has a solution.
Furthermore, the paper Arithmetic Expression Construction illustrates that if the Standard $\{-\}$ problem has a solution, it must be true that the partition problem has a solution. Since any solution for standard is in the same form as $(((p_1-n_2)-n_3) -\cdots -n_{|A2|})-(((n_1 -p_2) -p_3) -\cdots -p_{|A1|})$,  if it has a solution, the partition problem also has a solution.
Therefore, $(\mathbb{N},\{-\})$ Enforced Tree is weakly NP-hard.

Similar reasoning applies for other pairs of operations. For $\{+, *\}$, $\{+, /\}$, $\{-, *\}$, $\{-, /\}$,  $\{+, -,*\}$, $\{+,-,/\}$, $\{+, *,/\}$,$\{-, *,/\}$,$\{+, -, *, /\}$ -AEC-STD, all constructed instances must have a very particular partitioned structure described in Theorem $3.2$ in the paper \cite{alcock2020arithmetic}, so these enforced parenthesis problems must also be at least as hard as the standard problems. 

For the remaining constructions we present the following parenthsizations:

For $\{/\}$ -AEC-STD, the parenthesized structure is: \\ $(((\square\div \square) \div \cdots) \div \square) \div (((\square\div \square) \div \cdots) \div \square)$

For $\{+,-\}$ -AEC-STD, the parenthesized structure is: \\ $(((\square+\square)+ \cdots)+\square) - (((\square+ \square)+ \cdots) +\square)$.

For $\{*,/\}$ -AEC-STD, the parenthesized structure is: \\ $(((\square \times \square) \times \cdots) \times \square) \div (((\square \times \square) \times \cdots) \times \square)$.

\section{Conclusion and Open Problems}

We provided polynomial time algorithms or NP-completeness proofs for the Enforced Parenthesis and No Parenthesis variations of the arithmetic expression construction problem with all subsets of the typical four operations. However, many cases still have a gap between weak and strong NP-completeness. 

Variations of the problem inspired by Countdown are also unstudied from a computational complexity standpoint. One could investigate wheter only requiring a subset of the input numbers to be used changes the complexity, but we doubt this will be the case. We think it would be interesting to look into average case hardness for various arithmetic expression tree constructions. In addition, the process of deciding whether one wants a `large' or `small' number next seems like it could provide an interesting problem. This suggests something like deciding to draw random numbers from different distributions to try to maximize the probability of a tree existing.

One interesting avenue that is interesting to explore is expression tree construction with different operations. Min/Max, exponentiation, and modular arithmetic are common and will likely have similar characterizations. More exciting would be descriptions of properties of operations that allow these NP-hardenss reductions, or theorems similar to the rational function framework, to hold. 

Another important question to pursue is a careful consideration of the actual questions needing to be solved in the various machine learning architectures making use of expression tree prediction. Search and decision problems are not always the same complexity. Reuse of values is allowed in the constructions. In addition, one might know the expression trees have small size or low complexity leading to potential investigation of promise or parameterized problems. Or simple trees might be considered desirable properties leading to variations of arithmetic expression construction that want to minimize depth or tree size analogous to the Minimum Circuit Size Problem.

%
%
\bibliographystyle{splncs04}
\bibliography{references}

\clearpage

\appendix

\section{Related Problems}
\label{sec:related}

To show the NP-hardness of the variants of Arithmetic Expression Construction, we reduce from the following problems:


\defineproblemrefcomment{\ProductPartition}
{A multiset of positive integers $A = {a_1, a_2, \dots, a_n}$.}
{Can $A$ be partitioned into two subsets with equal product?}
{\cite{ProductPartition}.}
{Strongly NP-hard.}

\defineproblemrefcomment{\TPartition-3}
{A multiset of positive integers $A = {a_1, a_2, \dots, a_n}$, with $n$ a multiple of 3.}
{Can $A$ be partitioned into $n/3$ subsets with equal sum, where all subsets have size 3?}
{\cite{GareyJohnson}, problem SP15.}
{Strongly NP-hard, even when all subsets are required to have size 3 (\TPartitionEqual).}

\defineproblemrefcomment{\ProductPartitionEqual}
{A multiset of positive integers $A = {a_1, a_2, \dots, a_n}$.}
{Can $A$ be partitioned into two subsets with equal size $\frac{n}{2}$ and equal product?} {\cite{alcock2020arithmetic}.}
{Strongly NP-hard.}

\defineproblemrefcomment{\Partition}
{A multiset of positive integers $A = {a_1, a_2, \dots, a_n}$.}
{Can $A$ be partitioned into two subsets with equal sum?}
{\cite{GareyJohnson}, problem SP12.}
{Weakly NP-hard.}
\end{document}